\documentclass[12pt]{article}
\usepackage[utf8]{inputenc}

\usepackage{amsthm}
\usepackage{amsfonts}
\usepackage{amsmath}
\usepackage{txfonts}
\usepackage{url}
\usepackage{tikz}
\usetikzlibrary{patterns,shapes,calc}
\newtheorem{theorem}{Theorem}

\newtheorem{lemma}[theorem]{Lemma}
\newtheorem{corollary}[theorem]{Corollary}

\newcommand{\GG}{\mathcal{G}}
\newcommand{\PP}{\mathcal{P}}

\title{Independent sets near the lower bound in bounded degree graphs}

\author{Zdeněk Dvořák\thanks{Charles University, Prague, Czech Republic.
E-mail: {\tt rakdver@iuuk.mff.cuni.cz}.  Supported by (FP7/2007-2013)/ERC Consolidator grant LBCAD no. 616787.}\and
Bernard Lidický\thanks{Iowa State University, Ames IA, USA. E-mail: {\tt lidicky@iasate.edu}
Supported by NSF grant DMS-1600390.}}

\begin{document}
\maketitle

\begin{abstract}
By Brook's Theorem, every $n$-vertex graph of maximum degree at most $\Delta\ge 3$
and clique number at most $\Delta$ is $\Delta$-colorable, and thus it has an independent
set of size at least $n/\Delta$.  We give an approximate characterization of graphs with independence
number close to this bound, and use it to show that the problem of deciding whether such a graph
has an indepdendent set of size at least $n/\Delta+k$ has a kernel of size $O(k)$.
\end{abstract}

\section{Introduction}

Let $\Delta\ge 3$ be an integer and let $G$ be an $n$-vertex graph of maximum degree at most
$\Delta$ that does not contain a clique of size $\Delta+1$.
By Brooks' Theorem~\cite{brooks1941colouring} $G$ is $\Delta$-colorable, and thus the size of the largest independent set in $G$
(which we denote by $\alpha(G)$) is at least $n/\Delta$.  This bound is tight, as evidenced by graphs obtained
from a disjoint union of cliques of size $\Delta$ by adding a matching.  However, if the bound on the size of the largest
clique (which we denote by $\omega(G)$) is tightened, further improvements are possible.

Albertson, Bollobás and Tucker~\cite{AlbBolTuc} proved that
if a connected graph $G$ of maximum degree $\Delta$ satisfies $\omega(G)\le\Delta-1$, then $\alpha(G)>n/\Delta$, unless
$\Delta=4$ and $G=C_8^2$ or $\Delta=5$ and $G=C_5\boxtimes K_2$ (see Figure~\ref{fig-tight}(a) and (b)).
Under the same assumptions, the bound was further strengthened by King, Lu and Peng~\cite{kinglupeng},
who proved that $G$ has fractional chromatic number at most $\Delta-\frac{2}{67}$, and thus $\alpha(G)\ge \frac{n}{\Delta-\frac{2}{67}}$;
for particular values of $\Delta$, this bound has been further improved in~\cite{Sta79,thoheck,fracsub,kinged}, and Fajtlowicz~\cite{fajtlowicz1978size} gave better bounds
for graphs with smaller clique number.

In this paper, we study the case when $\omega(G)\le \Delta$ in more detail.  How do graphs with largest independent
set close to $n/\Delta$ look like?  An infinite class of examples can be obtained by generalizing the construction from the
first paragraph: if $G$ contains a subgraph $H$ isomorphic to the disjoint union of cliques of size $\Delta$ and $|V(G)\setminus V(H)|\le k$, then clearly
$\alpha(G)\le |V(H)|/\Delta+|V(G)\setminus V(H)|\le n/\Delta +k$.  For $\Delta\ge 6$, we prove an approximate converse to this
statement: If $\max(\Delta(G),\omega(G))\le \Delta$ and $\alpha(G)<n/\Delta+k$, then $G$ contains a subgraph $H$ isomorphic to the disjoint union of cliques of size $\Delta$
such that $|V(G)\setminus V(H)|<34\Delta^2k$.

For $\Delta\in \{3,4,5\}$, a more technical statement is needed.  We say that a graph $H$ is \emph{$\Delta$-tight} if
\begin{itemize}
\item $H$ is a clique on $\Delta$ vertices; or
\item $\Delta=5$ and $H=C_5 \boxtimes K_2$ (Figure~\ref{fig-tight}(a)); or
\item $\Delta=4$ and $H=C_8^2$ (Figure~\ref{fig-tight}(b)); or
\item $\Delta=4$ and $H$ is the graph depicted in Figure~\ref{fig-tight}(c), which we call an \emph{extended clique}, and
its vertices of degree $3$ are its \emph{attachments}; or,
\item $\Delta=4$ and $H$ is the graph depicted in Figure~\ref{fig-tight}(d), which we call an \emph{extended double-clique}; or
\item $\Delta=3$ and $H$ is one of the graphs depicted in Figure~\ref{fig-tight}(e), which we call \emph{diamond necklaces}; or
\item $\Delta=3$ and $H$ is one of the graphs depicted in Figure~\ref{fig-tight}(f), which we call \emph{Havel necklaces}; or,
\item $\Delta=3$ and $H$ is the graph depicted in Figure~\ref{fig-tight}(g), which we call a \emph{triangle-dominated $6$-cycle}.
\end{itemize}

\begin{figure}
\begin{center}
\tikzset{insep/.style={inner sep=1.7pt, outer sep=0pt, circle, fill},
 noin/.style={inner sep=0pt, outer sep=0pt, circle, fill},}
\begin{tikzpicture}[scale=1.2] 
\begin{scope}[xshift=0cm, yshift = 0cm]  
\draw
\foreach \x in {0,1,...,4}{(90+72*\x:1)node[insep](x\x){} -- (90+72*\x:0.5)node[insep](y\x){}}
(x0)--(x1)--(x2)--(x3)--(x4)--(x0)
(y0)--(y1)--(y2)--(y3)--(y4)--(y0)
(x0)--(y1)--(x2)--(y3)--(x4)--(y0)--(x1)--(y2)--(x3)--(y4)--(x0)
(0,-1.5) node{(a)}
;
\end{scope} 
\begin{scope}[xshift=3cm, yshift = 0cm]   
\draw
\foreach \x in {0,1,...,7}{(45*\x:1)node[insep](x\x){}}
(x0)--(x1)--(x2)--(x3)--(x4)--(x5)--(x6)--(x7)--(x0)
(x0)--(x2)--(x4)--(x6)--(x0) (x1)--(x3)--(x5)--(x7)--(x1)
(0,-1.5) node{(b)}
;
\end{scope}
\begin{scope}[xshift=6cm, yshift = 0cm]   
\draw
\foreach \x in {0,1,2,3}{(45+90*\x:0.5)node[insep](x\x){}}
(0,1) node[insep](e1){}
(0,-1) node[insep](e2){}
(1,0) node[insep](z1){}
(-1,0) node[insep](z2){}
(x0)--(x1)--(x2)--(x3)--(x1)--(z2)--(x2)--(x0)--(x3)--(z1)--(e1)--(z2)--(e2)--(z1)--(x0) (e1)--(e2)
(0,-1.5) node{(c)}
;
\end{scope}

\begin{scope}[xshift=9cm, yshift = 0cm]   
\draw
\foreach \x in {0,1,2,3}{(0,0.7)++(45+90*\x:0.5)node[insep](x\x){}}
(x0)--(x1)--(x2)--(x3)--(x0)--(x2)(x1)--(x3)
(1,0.7) node[insep](z1){}
(-1,0.7) node[insep](z2){}
(x0)--(z1)--(x3)(x2)--(z2)--(x1)
\foreach \x in {0,1,2,3}{(0,-0.7)++(45+90*\x:0.5)node[insep](x\x){}}
(x0)--(x1)--(x2)--(x3)--(x0)--(x2)(x1)--(x3)
(1,-0.7) node[insep](y1){}
(-1,-0.7) node[insep](y2){}
(x0)--(y1)--(x3)(x2)--(y2)--(x1)
(y1)--(z1) (y2)--(z2)
(y1) to[out=100,in=280] (z2)
(y2) to[out=80,in=260] (z1)
(0,-1.5) node{(d)}
;
\end{scope}
\end{tikzpicture}
\vskip 1em
\begin{tikzpicture}[scale=1.2] 
\begin{scope}[xshift=0cm, yshift = -3cm]   
\def\e{0.3}
\draw[rotate=90]
(0,\e) node[insep](x1){} -- (0,-\e) node[insep](x2){}
(\e,0) node[insep](y1){} (-\e,0) node[insep](y2){} 
(x1)--(y1)--(x2)--(y2)--(x1);
\begin{scope}[xshift=1.5cm,yshift=0cm]
\draw[rotate=90]
(0,\e) node[insep](a1){} -- (0,-\e) node[insep](a2){}
(\e,0) node[insep](b1){} (-\e,0) node[insep](b2){} 
(a1)--(b1)--(a2)--(b2)--(a1);
\end{scope}
\begin{scope}[xshift=0.75cm,yshift=-0.75cm]
\draw[rotate=0]
(0,\e) node[insep](x1){} -- (0,-\e) node[insep](x2){}
(\e,0) node[insep](z1){} (-\e,0) node[insep](z2){} 
(x1)--(z1)--(x2)--(z2)--(x1);
\end{scope}
\draw(z2)--(y2)  (y1) -- (b1) (b2)--(z1);
\begin{scope}[xshift=3cm,yshift=0cm]
\draw[rotate=90]
(0,\e) node[insep](x1){} -- (0,-\e) node[insep](x2){}
(\e,0) node[insep](y1){} (-\e,0) node[insep](y2){} 
(x1)--(y1)--(x2)--(y2)--(x1);
\begin{scope}[xshift=1.5cm,yshift=0cm]
\draw[rotate=90]
(0,\e) node[insep](a1){} -- (0,-\e) node[insep](a2){}
(\e,0) node[insep](b1){} (-\e,0) node[insep](b2){} 
(a1)--(b1)--(a2)--(b2)--(a1);
\end{scope}
\begin{scope}[xshift=0.75cm,yshift=-0.75cm]
\draw[rotate=0]
(0,0) node[insep](z2){};
\end{scope}
\draw(z2)--(y2)  (y1) -- (b1) (b2)--(z2);
\end{scope}
\begin{scope}[xshift=6cm,yshift=0cm]
\draw (0,0) node[insep](y){};
\begin{scope}[xshift=1.5cm,yshift=0cm]
\draw (0,0) node[insep](b){};
\end{scope}
\begin{scope}[xshift=0.75cm,yshift=-0.75cm]
\draw[rotate=0]
(0,\e) node[insep](x1){} -- (0,-\e) node[insep](x2){}
(\e,0) node[insep](z1){} (-\e,0) node[insep](z2){} 
(x1)--(z1)--(x2)--(z2)--(x1);
\end{scope}
\draw(z2)--(y) -- (b) --(z1);
\end{scope}
\draw (3.75,-1.5)node{(e)};
\end{scope}
\end{tikzpicture}
\\
\begin{tikzpicture}[scale=1.5]
\begin{scope}
\draw (0.15,0) node[insep](a){}
--+(45:0.3) node[insep](b){}
--+(-45:0.3) node[insep](c){}--(a)
--++(-0.3,0) node[insep](x){}
--+(135:0.3) node[insep](y){}
--+(225:0.3) node[insep](z){}--(x)
(0,0.7) node[insep](k){}
(0,-0.7) node[insep](l){}
(b)--(k)--(y)
(c)--(l)--(z)
(1,0) node[insep](q){}
(k)--(q)--(l)
;
\end{scope}
\begin{scope}[xshift=3cm]
\draw (0.15,0) node[insep](a){}
--+(45:0.3) node[insep](b){}
--+(-45:0.3) node[insep](c){}--(a)
--++(-0.3,0) node[insep](x){}
--+(135:0.3) node[insep](y){}
--+(225:0.3) node[insep](z){}--(x)
(0,0.7) node[insep](k){}
(0,-0.7) node[insep](l){}
(b)--(k)--(y)
(c)--(l)--(z)
(1,0) node[insep](q1){}
--++(45:0.3) node[insep](q2){}
--++(-45:0.3) node[insep](q3){}
--++(45:-0.3) node[insep](q4){}--(q1)--(q3)
(k)--(q2) (l)--(q4)
;
\end{scope}
\draw(1.8,-1) node{(f)};
\end{tikzpicture}
\\
\begin{tikzpicture}[scale=0.8]
\draw
(-2,0) node[insep](z2){}
(-3,-1) node[insep](z1){}
(-3,1) node[insep](z3){}
(z1)--(z2)--(z3)--(z1)
(2,0) node[insep](w2){}
(3,-1) node[insep](w1){}
(3,1) node[insep](w3){}
(w1)--(w2)--(w3)--(w1)
\foreach \x in {1,2,3}{
(-0.6,-2+\x) node[insep](x\x){}--(z\x)
(0.6,-2+\x) node[insep](y\x){}--(w\x)
}
(x1) -- (y1)--(x2)--(y3)--(x3)--(y2)--(x1)
;
\draw(0,-2) node{(g)};
\end{tikzpicture}
\end{center}
\caption{Non-clique $\Delta$-tight graphs.}\label{fig-tight}
\end{figure}
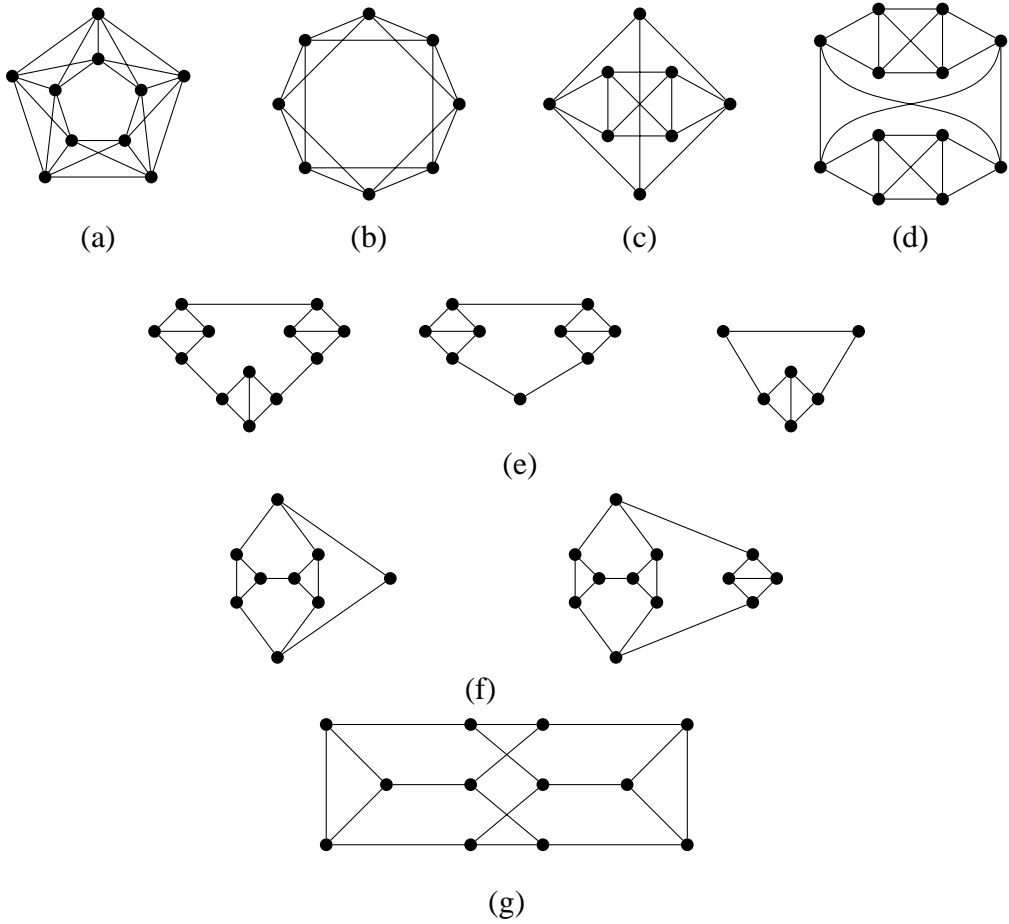

Observe that if $H$ is $\Delta$-tight, then $\alpha(H)=|V(H)|/\Delta$.  We say that a graph $G_0$ is 
\emph{$\Delta$-tightly partitioned} if there exists a partition of the vertices of $G_0$ such that
each part induces a $\Delta$-tight subgraph of $G_0$.  Clearly, a $\Delta$-tightly partitioned graph $G_0$ satisfies $\alpha(G_0)\le |V(G_0)|/\Delta$.
Our main result now can be stated as follows.
\begin{theorem}\label{thm-approx}
Let $\Delta\ge 3$ and $k\ge 0$ be integers, and let $G$ be an $n$-vertex graph with $\max(\Delta(G),\omega(G))\le \Delta$.
If $\alpha(G)<n/\Delta+k$, then there exists a set $X\subseteq V(G)$ of size less than $34\Delta^2k$ such that
$G-X$ is $\Delta$-tightly partitioned.
\end{theorem}

Note that the set $X$ of Theorem~\ref{thm-approx} can be found in time $O(\Delta^2n)$ without specifying $k$ in
advance (cf. a stronger statement in Lemma~\ref{lemma-decomp} below),
and that every graph containing such a set $X$ satisfies $\alpha(G)\le |V(G)|/\Delta+|X|$.
Hence, we obtain the following algorithmic consequence.
\begin{corollary}\label{cor-approx}
Let $\Delta\ge 3$ be an integer, and let $G$ be an $n$-vertex graph with $\max(\Delta(G),\omega(G))\le \Delta$.
The difference $\alpha(G)-n/\Delta$ can be approximated up to factor $34\Delta^2$ in time $O(\Delta^2n)$.
\end{corollary}

Another application of Theorem~\ref{thm-approx} (or more precisely, its refinement Lemma~\ref{lemma-decomp}) concerns fixed-parameter
tractability.  An algorithmic problem is called \emph{fixed-parameter tractable} with respect to a parameter $p$ if
there exists a computable function $f$, a polynomial $q$, and an algorithm that solves each input instance $Z$ in time $f(p(Z))q(|Z|)$.
This notion has been influential in the area of computational complexity, giving a plausible approach towards many otherwise intractable
problems~\cite{downey2012parameterized,Niedermeier2006,cygan2015parameterized}.

A popular choice of the parameter is the value of the solution; i.e., such fixed-parameter tractability results show that
the solution to the problem can be found quickly if its value is small.  However, in the case of the problem of finding the largest independent set
when restricted to some class of sparse graphs, this parameterization makes little sense---the problem is fixed-parameter tractable
for the trivial (and unhelpful) reason that all large graphs in the class have large independent sets.  In this setting, parameterization
by the excess of the size of the largest independent set over the lower bound for the independence number in the class is more reasonable.

Let $\GG$ be a class of graphs, and let $f(n)=\min\{\alpha(G):G\in \GG, |V(G)|=n\}$.  Let \emph{$(\GG,\alpha)$-ATLB} (Above Tight Lower Bound)
denote the algorithmic problem of deciding, for an $n$-vertex graph $G\in \GG$ and an integer $k\ge 0$, whether
$\alpha(G)\ge f(n)+k$.  For specific graph classes $\GG$, we are interested in the fixed-parameter tractability of this problem
when parameterized by $k$, i.e., in finding algorithms for this problem with time complexity $f(k)\text{poly}(n)$.

The best-known case of this problem concerns the class $\PP$ of planar graphs.  By the Four Color Theorem~\cite{AppHak1,AppHakKoc},
all $n$-vertex planar graphs have independent sets of size at least $\lceil n/4\rceil$, and this lower bound is tight.
However, there is not even a polynomial-time algorithm known to decide whether $\alpha(G)>n/4$ for an $n$-vertex planar graph $G$,
and consequently the complexity of $(\PP,\alpha)$-ATLB is wide open~\cite{Niedermeier2006,BodlaenderEtAl2008,FellowsEtAl2012}.
The variant of the problem for planar triangle-free graphs was solved by Dvořák and Mnich~\cite{dmnich}.
Relevantly to the current paper, they also considered the case of planar graphs of maximum degree $4$ and showed it to be
fixed-parameter tractable~\cite{mnich2016large,dmnichfull}; for both problems, they obtained algorithms with time complexity $2^{O(\sqrt{k})}n$.

As a consequence of our theory, we strengthen the last mentioned result in multiple directions.  Firstly, we can drop the assumption
of planarity and relax the condition on the maximum degree to an arbitrary integer $\Delta\ge 3$ (assuming that cliques of size $\Delta+1$
are forbidden---let $\GG_\Delta$ denote the class of such graphs).  Secondly, we obtain a \emph{linear-size kernel} for the problem, i.e., we show that in polynomial time (specifically,
$O(\Delta^2n)$), each instance of the problem can be reduced to an instance of size $O(\Delta^3k)$.
\begin{corollary}\label{cor-kernel}
There exists an algorithm with time complexity $O(\Delta^2 n)$ that, given as an input 
an integer $\Delta\ge 3$, an $n$-vertex graph $G$ with $\max(\Delta(G),\omega(G))\le\Delta$, and an integer $k\ge 0$,
returns an induced subgraph $G_0$ of $G$ with $n_0\le 114\Delta^3k$ vertices such that $\alpha(G)\ge n/\Delta+k$ if and only if $\alpha(G_0)\ge n_0/\Delta+k$.
\end{corollary}
Such an instance can be solved by brute force, leading to a $2^{O(\Delta^3k)}+O(\Delta^2n)$ algorithm for $(\GG_\Delta,\alpha)$-ATLB.
Furthermore, as Robertson et al.~\cite{quickly} showed, an $m$-vertex planar graph has tree-width $O(\sqrt{m})$,
and thus its largest independent set can be found in time $2^{O(\sqrt{m})}$.  Hence, for the case of planar graphs of maximum degree $4$
previously considered in~\cite{mnich2016large,dmnichfull}, we can improve the time complexity to $2^{O(\sqrt{k})}+O(n)$.

After introducing graph-theoretic results on sizes of independent sets in Section~\ref{sec-uncov}, we prove Theorem~\ref{thm-approx} and Corollaries~\ref{cor-approx} and \ref{cor-kernel}
in Section~\ref{sec-mainproofs}.

\section{Many vertices not covered by $\Delta$-tight subgraphs}\label{sec-uncov}

We say that a vertex $v\in V(G)$ is \emph{$\Delta$-free} if $v$ is not contained in any $\Delta$-tight induced subgraph of $G$.
As a special case, Theorem~\ref{thm-approx} implies that if $G$ has many $\Delta$-free vertices, then its largest independent
set is much larger than the lower bound $n/\Delta$.  In this section, we prove this special case in Lemmas~\ref{lemma-free} and \ref{lemma-free3}.
Let us start by a simple observation.

\begin{lemma}\label{lemma-reducl}
Let $\Delta>0$ be an integer and let $G$ be a graph with $\Delta(G)\le \Delta$,
and let $K$ be a clique of size $\Delta$ in $G$.  Consider one of the following situations:
\begin{itemize}
\item $K$ contains a vertex $v$ whose degree in $G$ is $\Delta-1$, and $G_0=G-V(K)$, or
\item $V(G)\setminus V(K)$ contains two adjacent vertices $u'$ and $v'$ with neighbors in $K$, and $G_0=G-V(K)$, or
\item $V(G)\setminus V(K)$ contains two distinct non-adjacent vertices $u'$ and $v'$ with neighbors in $K$, and $G_0=G-V(K)+u'v'$.
\end{itemize}
Then $\alpha(G)\ge\alpha(G_0)+1$.
\end{lemma}
\begin{proof}
Since $\Delta(G)\le \Delta$, every vertex of $K$ has at most one neighbor not in $K$.
In the last two cases, let $u$ and $v$ be the neighbors of $u'$ and $v'$ in $K$, respectively.
Consider an independent set $S$ of $G_0$ of size $\alpha(G_0)$.

In the first case, $v$ has no neighbor in $S$, as all neigbors of $v$ belong to $K$.
In the last two cases, since $u'v'\in E(G_0)$, we can by symmetry assume that
$v'\not\in S$.  Hence, $S\cup\{v\}$ is in both cases an independent set in $G$.
\end{proof}

We say that distinct vertices $u$ and $v$ are \emph{$\Delta$-adjacent} in a graph $G$ if $uv\not\in E(G)$ and $G$ contains
an induced subgraph $H_0$ with $u,v\in V(H_0)$ such that $H_0+uv$ is $\Delta$-tight.

For a vertex $v$, let $N_G[v]$ denote the set consisting of $v$ and of the neighbors of $v$ in $G$,
and for a set $S\subseteq V(G)$, let $N_G[S]=\bigcup_{v\in S} N_G[v]$.
We say that a set $Z\subseteq V(G)$ is \emph{$\Delta$-profitable} if $G[Z]$ contains an independent
set $S$ such that $N_G[S]\subseteq Z$ and $|Z|\le \Delta|S|-1$.  Thus, any independent set of $G-Z$
can be extended to an independent set in $G$ by adding $S$, and this way the size of the independent
set increases by slightly more than $|Z|/\Delta$.

\begin{lemma}\label{lemma-noadj}
Let $\Delta\ge 4$ be an integer and let $G$ be a graph with $\Delta(G)\le \Delta$.
Let $S=\{v_1,v_2,v_3\}$ be a subset of vertices of $G$.
If $G$ contains no $\Delta$-profitable set of size at most $\Delta+10$, then the vertices of $S$ are not
pairwise $\Delta$-adjacent.
\end{lemma}
\begin{proof}
Suppose for a contradiction that vertices of $S$ are pairwise $\Delta$-adjacent, and in particular $S$ is an independent set.
For $1\le i<j\le 3$, let $H_{ij}$ be the induced subgraph of $G$ showing $\Delta$-adjacency of $v_i$ and $v_j$.

Suppose first that all integers $i$ and $j$ such that $1\le i<j\le 3$ satisfy $|N_G[v_i]\cap N_G[v_j]|\ge \Delta-2$.
Let $Z=N_G[S]$; by the principle of inclusion and exclusion we have
\begin{align*}
|Z|&=\sum_{i=1}^3 |N_G[v_i]|-\sum_{1\le i<j\le 3} |N_G[v_i]\cap N_G[v_j]|+\left|\bigcap_{i=1}^3 N_G[v_i]\right|\\
&\le \sum_{i=1}^3 |N_G[v_i]|-\sum_{i=2}^3 |N_G[v_1]\cap N_G[v_i]|\\
&\le 3(\Delta+1)-2(\Delta-2)=\Delta+7\le 3\Delta-1=\Delta|S|-1.
\end{align*}
It follows that $Z$ is $\Delta$-profitable, which is a contradiction.

Hence, we can assume that say $|N_G[v_1]\cap N_G[v_2]|\le \Delta-3$, and in particular $H'_{12}=H_{12}+v_1v_2$ is
not a clique of size $\Delta$. Consequently, $\Delta\in\{4,5\}$.  Note that $H_{12}'$ is not an extended clique with attachments
$v_1$ and $v_2$, as otherwise we would have $|N_G[v_1]\cap N_G[v_2]|=2=\Delta-2$.

If $H_{12}'$ is not an extended double-clique, then $H_{12}'$ is $C_8^2$, $C_5 \boxtimes K_2$, or the extended clique,
and $|V(H_{12})|=2\Delta$.  A vertex of $V(H_{12})$ may only have neighbors outside of $V(H_{12})$ in $G$ if
its degree in $H_{12}'$ is less than $\Delta$, or if it is one of $v_1$ and $v_2$.
Using this observation, a case analysis shows that $H_{12}$ contains an independent set $S'$ of size three including $v_1$ and $v_2$,
such that $|N_G[S']\setminus V(H_{12})|\le 7-\Delta$.  Letting $Z'=N_G[S']\cup V(H_{12})$, we have
$|Z'|\le 2\Delta+(7-\Delta)=\Delta+7\le 3\Delta-1=\Delta|S'|-1$.
Hence, $Z'$ is $\Delta$-profitable, which is a contradiction.

If $H_{12}'$ is an extended double-clique (and $\Delta=4$), then $|V(H_{12})|=12$ and $H_{12}$ contains an independent set $S'$ of size
$4$ (including $v_1$ and $v_2$) such that $|N_G[S']\setminus V(H_{12})|\le 2$.  Letting $Z'=N_G[S']\cup V(H_{12})$, we have
$|Z'|\le 14=\Delta+10<4\Delta-1=\Delta|S'|-1$.  Again, $Z'$ is $\Delta$-profitable, which is a contradiction.
\end{proof}

King, Lu and Peng~\cite{kinglupeng} proved that every graph $G$ of maximum degree at most $\Delta\ge 4$ and with no
$\Delta$-tight induced subgraph has fractional chromatic number at most $\Delta-2/67$.
This gives a lower bound on the independence number of graphs whose vertices are all $\Delta$-free.
\begin{corollary}\label{cor-king}
Let $\Delta\ge 4$ be an integer, and let $G$ be an $n$-vertex graph of maximum degree at most $\Delta$.
If $G$ contains no $\Delta$-tight induced subgraph, then $\alpha(G)\ge \frac{n}{\Delta} + \frac{1}{34\Delta^2}n$.
\end{corollary}
\begin{proof}
Since $G$ has fractional chromatic number at most $\frac{67\Delta-2}{67}$, we have
\begin{align*}
\alpha(G)&\ge \frac{67n}{67\Delta-2}=\frac{(67-2/\Delta)n+2n/\Delta}{67\Delta-2}\\
&=\frac{n}{\Delta}+\frac{2n}{\Delta(67\Delta-2)}\ge \frac{n}{\Delta}+\frac{n}{34\Delta^2}.
\end{align*}
\end{proof}

We now extend this result to graphs with only some $\Delta$-free vertices.

\begin{lemma}\label{lemma-free}
Let $\Delta\ge 4$ be an integer, and let $G$ be an $n$-vertex graph with $\max(\Delta(G),\omega(G))\le \Delta$.
If $G$ contains at least $m$ $\Delta$-free vertices, then $\alpha(G)\ge \frac{n}{\Delta} + \frac{1}{34\Delta^2}m$.
\end{lemma}
\begin{proof}
We proceed by induction, and thus we can assume that the claim holds for all graphs with less than $n$ vertices.

If $H$ is a $\Delta$-tight induced subgraph of $G$ that is not a clique, then observe that $\alpha(H)=t$ for some $t\in\{2,3\}$
and that $H$ has an independent set of size $t$ whose closed neighborhood in $G$ is contained in $V(H)$.
Consequently $\alpha(G)=\alpha(G-V(H))+t$.  Furthermore, $|V(H)|=t\Delta$ and $G-V(H)$ has
at least $m$ $\Delta$-free vertices, and thus
$\alpha(G)=\alpha(G-V(H))+t\ge \frac{n-t\Delta}{\Delta}+\frac{1}{34\Delta^2}m+t=\frac{n}{\Delta} + \frac{1}{34\Delta^2}m$
by the induction hypothesis.
Hence, we can without loss of generality assume that the only $\Delta$-tight induced subgraphs of $G$ are cliques of size $\Delta$.

Suppose that $G$ contains a $\Delta$-profitable set $Z$ of size at most $\Delta+10$, and let $S$ be the corresponding independent set.
Note that the number of $\Delta$-free vertices of $G-Z$ is at least $m-|Z|$, and thus
\begin{align*}
\alpha(G)&\ge \alpha(G-Z)+|S|\\
&\ge \frac{n-|Z|}{\Delta}+\frac{1}{34\Delta^2}(m-|Z|)+|S|\\
&=\frac{n}{\Delta}+\frac{1}{34\Delta^2}m+\frac{\Delta|S|-|Z|-|Z|/(34\Delta)}{\Delta}.
\end{align*}
Since $Z$ is $\Delta$-profitable, we have $\Delta|S|-|Z|\ge 1>\frac{|Z|}{34\Delta}$,
and thus the inequality we seek holds.

If $\omega(G)<\Delta$, then all vertices of $G$ are $\Delta$-free and the claim follows from Corollary~\ref{cor-king}.
It remains to consider the case that $G$ contains a clique $K$ of size $\Delta$, but does not contain any $\Delta$-profitable sets
with at most $\Delta+10$ vertices.
Let $S$ be the set of vertices outside $K$ with a neighbor in $K$.  If $K$ contains a vertex whose degree in $G$ is
$\Delta-1$, or if $S$ is not an independent set, then note that $G-V(K)$ contains at least $m$ $\Delta$-free vertices,
and by Lemma~\ref{lemma-reducl} and the induction hypothesis, we have
$\alpha(G)\ge \alpha(G-V(K))+1\ge \frac{n-\Delta}{\Delta}+\frac{1}{34\Delta^2}m+1=\frac{n}{\Delta}+\frac{1}{34\Delta^2}m$.
Hence, suppose that $S$ is an independent set and each vertex of $K$ has a neighbor in $S$ (and thus there are precisely $\Delta$
edges between $V(K)$ and $S$).

Since $\omega(G)\le\Delta$, we have $|S|\ge 2$.  If some two vertices $u,v\in S$ are not $\Delta$-adjacent in $G-V(K)$, then
$G_0=G-V(K)+uv$ has at least $m$ $\Delta$-free vertices, and the inequality follows again by Lemma~\ref{lemma-reducl}
and the induction hypothesis applied to $G_0$.  Hence, suppose that the vertices of $S$ are pairwise $\Delta$-adjacent.
By Lemma~\ref{lemma-noadj}, we have $|S|=2$.  Since the vertices of $S$ are $\Delta$-adjacent in $G-V(K)$, they have degree
at least $\Delta-2$ in $G-V(K)$.  Since their degree in $G$ is at most $\Delta$, the number of edges between $V(K)$ and $S$ is at most $4$.  Therefore, there are
exactly $4=\Delta$ edges between $V(K)$ and $S$ and both vertices $u$ and $v$ of $S$ have degree exactly $\Delta-2$ in $G-V(K)$.

Let $H$ be the $4$-tight induced subgraph in $G-V(K)+uv$ containing the edge $uv$.  Since both $u$ and $v$
have degree at most $\Delta-1$ in $H$, we conclude that $H$ is either the clique on $4$ vertices or the extended clique
with attachments $u$ and $v$.  In the former case, $H'=G[V(K)\cup V(H)]$ is an extended clique
with the common neighbors of $u$ and $v$ in $H$ as attachments.  In the latter case, $H'=G[V(K)\cup V(H)]$ is
an extended double-clique.  In both cases, $H'$ is a $4$-tight induced subgraph of $G$ distinct from a clique, which is a contradiction.
\end{proof}

Next, we prove a similar claim for graphs of maximum degree at most $3$.  Staton~\cite{Sta79}
proved that every subcubic triangle-free $n$-vertex graph has independence number at least $5n/14$.
In particular, we have the following.

\begin{corollary}\label{cor-dsv}
Let $G$ be an $n$-vertex graph of maximum degree at most $3$.
If $G$ contains no 3-tight induced subgraph, then $\alpha(G)\ge \frac{n}{3} + \frac{1}{42}n$.
\end{corollary}

We aim to generalize this result to graphs containing 3-tight subgraphs.
In comparison with Lemma~\ref{lemma-free}, we will actually need to strengthen
the claim even more and gain also for some of the vertices contained in
triangles.

A \emph{diamond} is an induced subgraph isomorphic to a clique on $4$ vertices minus one edge.
A \emph{necklace} is either a diamond necklace or a Havel necklace.
A diamond in a graph $G$ is \emph{free} if it is not contained in an induced subgraph of $G$ isomorphic to a necklace.
For a subcubic graph $G$, let $m(G)=m_1+m_2$, where $m_1$ is the number of $3$-free vertices of $G$ and $m_2$ is the
number of free diamonds of $G$.  Note that the same bound was obtained in~\cite{kanjzhang} under more restrictive
assumptions (only for diamond-free graphs without Havel necklaces and triangle-dominated $6$-cycles).

\begin{lemma}\label{lemma-free3}
If $G$ is an $n$-vertex graph with $\max(\Delta(G),\omega(G))\le 3$,
then $\alpha(G)\ge \frac{n}{3} + \frac{1}{42}m(G)$.
\end{lemma}
\begin{proof}
We proceed by induction, and thus we can assume that the claim holds for all graphs with less than $n$ vertices.

If $H$ is an induced subgraph of $G$ isomorphic to a necklace or a triangle-dominated $6$-cycle, then observe that $t=\alpha(H)\in\{2,3,4\}$
and $H$ has an independent set of size $t$ whose closed neighborhood in $G$ is contained in $V(H)$.
Therefore, $\alpha(G)=\alpha(G-V(H))+t$.  Furthermore, no vertex of $V(H)$ is $3$-free or contained in a free diamond of $G$,
and thus $m(G-V(H))\ge m(G)$.
Observe that $|V(H)|=3t$, and hence $\alpha(G)=\alpha(G-V(H))+t\ge \frac{n-3t}{3}+\frac{1}{42}m(G-V(H))+t\ge \frac{n}{3} + \frac{1}{42}m(G)$
by the induction hypothesis.
Hence, we can without loss of generality assume that the only $3$-tight induced subgraphs of $G$ are triangles.

Suppose that $H$ is a diamond in $G$, which is necessarily free.  Let $G'$ be obtained from $G$ by contracting all vertices of $H$ into a single vertex $v$.
Observe that $v$ has degree at most 2 in $G'$, and thus $G'$ is $K_4$-free.  If $v$ were contained in either a necklace or a triangle in $G'$,
then $H$ would be a part of a necklace in $G$, contradicting the conclusion of the previous paragraph.
Consequently, we replaced a free diamond $H$ of $G$ by a $3$-free vertex $v$ in $G'$, and thus $m(G')\ge m(G)$.
Moreover, for every independent set $S'$ in $G'$ there exists an independent set $S$ in $G$ such that $|S| = |S'|+1$
(if $v\in S'$, we can add the two non-adjacent vertices of $H$ to $S$ instead; if $v\not\in S'$, then we can
add one of the vertices whose degree in $H$ is $3$ to $S$). Hence,
$\alpha(G) \geq \alpha(G')+1\ge \frac{n-3}{3}+\frac{1}{42}m(G')+1\ge \frac{n}{3} + \frac{1}{42}m(G)$.
Therefore, we can without loss of generality assume that $G$ is diamond-free.

If $G$ is triangle-free, then $\alpha(G)\ge \frac{n}{3} + \frac{1}{42}n=\frac{n}{3} + \frac{1}{42}m(G)$ by Corollary~\ref{cor-dsv}.
Hence, suppose that $H$ is a triangle in $G$, with vertices $v_1$, $v_2$, and $v_3$.
If say $v_1$ has degree two in $G$, then
every independent set in $G-V(H)$ can be extend to an independent set in $G$
by including $v_1$. Moreover, none of the vertices in $V(H)$ is 3-free or in a free diamond,
which gives $m(G-V(H))\ge m(G)$.
Hence,
$\alpha(G) \geq \alpha(G - V(H)) + 1 \geq \frac{n-3}{3}+  \frac{1}{42}m(G-V(H)) + 1
\ge   \frac{n}{3}+  \frac{1}{42}m(G)$.
Therefore, we can without loss of generality assume that $V(H)$ contains only vertices of degree three.

For $i \in \{1,2,3\}$, let $u_i$ denote the neighbor of $v_i$ not in $V(H)$.
Since $G$ is diamond-free, $u_i$ and $u_j$ are distinct for all $1 \leq i < j \leq 3$. 
If $u_i$ is adjacent to $u_j$ for some $1 \leq i < j \leq 3$, then Lemma~\ref{lemma-reducl} and the induction hypothesis give
$\alpha(G) \geq \alpha(G - V(H)) + 1 \geq \frac{n-3}{3}+  \frac{1}{42}m(G-V(H)) + 1 \ge   \frac{n}{3}+  \frac{1}{42}m(G)$.
Hence, assume that $U=\{u_1,u_2,u_3\}$ is an independent set in $G$.

\begin{figure}
\begin{center}
\tikzset{vblack/.style={inner sep=1.7pt, outer sep=0pt, circle, fill},
 vwhite/.style={inner sep=1.7pt, outer sep=0pt, circle, fill=white,draw},}
\begin{tikzpicture}[scale=1.5]
\begin{scope}
\draw (0.15,0) node[vwhite](a){}
--+(45:0.3) node[vwhite](b){}
--+(-45:0.3) node[vblack](c){}--(a)
--++(-0.3,0) node[vwhite](x){}
--+(135:0.3) node[vwhite](y){}
--+(225:0.3) node[vblack](z){}--(x)
(0,0.7) node[vblack,label=above:$u_i$](k){}
(0,-0.7) node[vwhite](l){}
(b)--(k)--(y)
(c)--(l)--(z)
(1,-0.5) node[vblack,label=below:$u_j$](q){}
(q)--(l)
(1.5,0) node[vblack](x){}
--+(135:0.3) node[vwhite](y){}
--+(225:0.3) node[vwhite](z){}--(x)
(z)--(q)
(k)--(y)
(q) -- ++(0.3,0) node[vwhite]{}
(x) -- ++(0.3,0) node[vwhite]{}
;
\draw[dotted,very thick]
(k) to[bend left=15](q)
;
\end{scope}
\begin{scope}[xshift=4cm]
\draw (0.15,0) node[vblack](a){}
--+(45:0.3) node[vwhite](b){}
--+(-45:0.3) node[vwhite](c){}--(a)
--++(-0.3,0) node[vwhite](x){}
--+(135:0.3) node[vwhite](y){}
--+(225:0.3) node[vblack,label=below:$u_j$](z){}--(x)
(0,0.7) node[vblack](k){}
(0,-0.7) node[vblack,label=below:$u_i$](l){}
(b)--(k)--(y)
(c)--(l)
(1,0) node[vwhite](q){}
(k)--(q)--(l)
(-1,-0.5) node[vblack](a){}
--+(45:0.3) node[vwhite](b){}
--+(-45:0.3) node[vwhite](c){}--(a)
(c)--(l)
(z)--(b)
(a) -- ++(-0.3,0) node[vwhite]{}
;
\draw[dotted,very thick]
(l) to[bend left=0](z)
;
\end{scope}
\begin{scope}[xshift=7cm]
\draw (0.5,0) node[vblack,label=above:$u_i$](a){}
--+(45:0.3) node[vwhite](b){}
--+(-45:0.3) node[vwhite](c){}--(a)
(-0.5,0) node[vblack,label=above:$u_j$](x){}
--+(135:0.3) node[vwhite](y){}
--+(225:0.3) node[vwhite](z){}--(x)
(0,0.9) node[vblack,label=above:$ $](k){}
(0,-0.9) node[vblack](l){}
(b)--(k)--(y)
(c)--(l)--(z)
(1.2,0) node[vwhite](q){}
(k) to[bend left] (q) (q) to[bend left] (l)
(0,0.15) node[vblack](m){}
--+(-45:0.3) node[vwhite](n){}
--+(45:-0.3) node[vwhite](o){}--(m)
(a)--(n)
(x)--(o)
(m) -- ++(0,0.3) node[vwhite]{}
;
\draw[dotted,very thick]
(x) to[bend left=-50, looseness=1.5](a)
;
\end{scope}
\end{tikzpicture}\\[10pt]
\begin{tikzpicture}[scale=0.8]
\draw
(-2,0) node[vblack](z2){}
(-3,-1) node[vwhite](z1){}
(-3,1) node[vwhite](z3){}
(z1)--(z2)--(z3)--(z1)
(2,0) node[vwhite](w2){}
(3,-1) node[vwhite](w1){}
(3,1) node[vblack](w3){}
(w1)--(w2)--(w3)--(w1)
\foreach \x in {1,2,3}{
(-0.6,-2+\x) node[vwhite](x\x){}--(z\x)
(0.6,-2+\x) node[vblack](y\x){}
}
(y1)--(w1)
(y2)--(w2)
(x1) -- (y1)--(x2)--(y3)--(x3)--(y2)--(x1)
(w3) -- ++(0,0.7) node[vwhite]{}
(y3) -- ++(0,0.7) node[vwhite]{}
(y3) node[vblack,label=above left:$u_i$]{}
(w3) node[vblack,label=right:$u_j$]{}
;
\draw[dotted,very thick]
(w3)--(y3)
;
\end{tikzpicture}
\hskip 2em
\begin{tikzpicture}[scale=0.8]
\draw
(-2,0) node[vblack](z2){}
(-3,-1) node[vwhite](z1){}
(-3,1) node[vwhite](z3){}
(z1)--(z2)--(z3)--(z1)
(2,0) node[vblack](w2){}
(3,-1) node[vwhite](w1){}
(3,1) node[vwhite](w3){}
(w1)--(w2)--(w3)--(w1)
\foreach \x in {1,2,3}{
(-0.6,-2+\x) node[vblack](x\x){}--(z\x)
(0.6,-2+\x) node[vblack](y\x){}--(w\x)
}
(x2) node[vwhite]{}
(y2) node[vwhite]{}
(y1) node[vwhite]{}
(x1) -- (y1)--(x2)--(y3) (x3)--(y2)--(x1)
(x3) -- ++(0,0.7) node[vwhite]{}
(y3) -- ++(0,0.7) node[vwhite]{}
(x3) node[vblack,label=above left:$u_i$]{}
(y3) node[vblack,label=above right:$u_j$]{}
;
\draw[dotted,very thick]
(x3)--(y3)
;
\end{tikzpicture}
\end{center}
\caption{Subgraphs arising from a Havel necklace or a triangle-dominated $6$-cycle.}\label{fig-exthavel}
\end{figure}
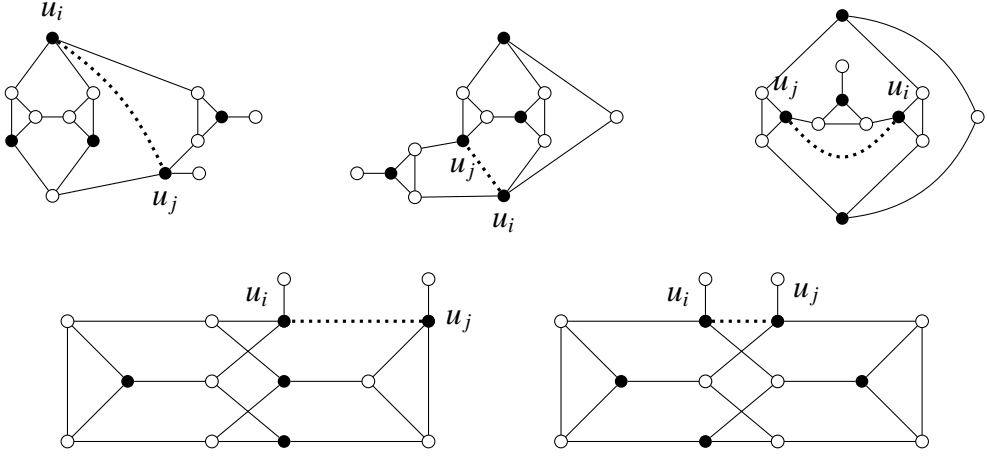

Suppose that $u_i$ and $u_j$ have no common neighbor for some $1 \leq i < j \leq 3$.
If $u_i$ is not $3$-adjacent to $u_j$ in $G$, then
let $G'$ be the graph obtained from $G-V(H)$ by adding the edge $u_iu_j$.
We have $m(G)\ge m(G')$, and
Lemma~\ref{lemma-reducl} together with the induction hypothesis give
$\alpha(G) \geq \alpha(G') + 1 \geq \frac{n-3}{3}+  \frac{1}{42}m(G') + 1 \ge \frac{n}{3}+  \frac{1}{42}m(G)$.
Hence, we can assume that $u_i$ and $u_j$ are $3$-adjacent in $G$.
Since $G$ is diamond-free and $u_i$ and $u_j$ do not have a common neighbor,
we conclude that $G$ contains an induced subgraph $H_0$ containing $u_i$ and $u_j$
such that $H_0+u_iu_j$ is isomorphic to the diamond-free Havel necklace or the triangle-dominated $6$-cycle,
and $G$ contains a subgraph $H_1$ isomorphic to a graph that is either depicted in Figure~\ref{fig-exthavel}
or obtained from one of the depicted graphs by identifying a vertex of degree $1$ with another vertex of degree at most $2$.
Note that $|V(H_1)|\le 14$, and as shown in the figure, $H_1$ contains an independent set of size $5$
whose closed neighborhood in $G$ is contained in $V(H_1)$.  Since at most $6$ vertices of $V(H_1)$
are $3$-free in $G$, we have $m(G-V(H_1))\ge m(G)-6$.  Together with the induction hypothesis,
we obtain
$\alpha(G) \geq \alpha(G-V(H_1)) + 5 \geq \frac{n-14}{3}+  \frac{1}{42}m(G-V(H_1)) + 5 \ge \frac{n+1}{3}+  \frac{1}{42}(m(G)-6)> \frac{n}{3}+  \frac{1}{42}m(G)$.

Therefore, we can assume that $u_i$ and $u_j$ have a common neighbor for all $1 \leq i < j \leq 3$.
There are two cases---either these common neighbors are pairwise different, or there
exists a common neighbor of all vertices of $U$.
We first consider the case that there exists a vertex $x$ adjacent to $u_1$, $u_2$, and $u_3$.
We distinguish two subcases.

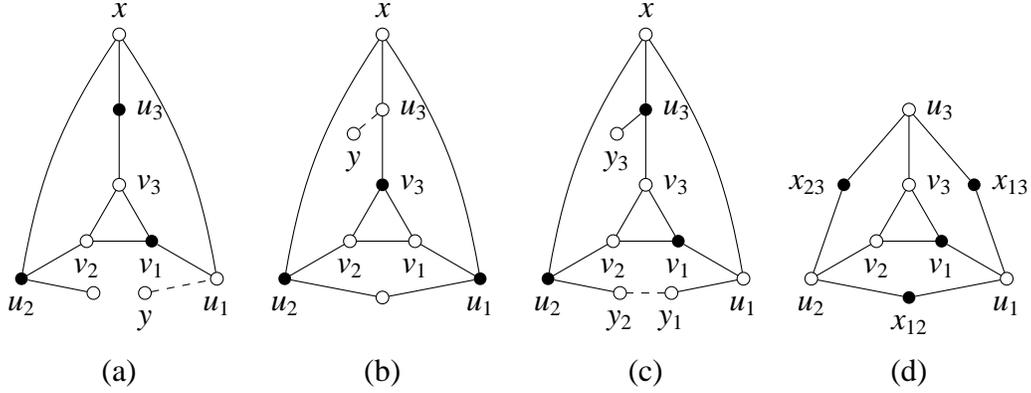
\begin{figure}
\begin{center}
\tikzset{vblack/.style={inner sep=1.7pt, outer sep=0pt, circle, fill},
 vwhite/.style={inner sep=1.7pt, outer sep=0pt, circle, fill=white,draw},}
\begin{tikzpicture}[scale=0.5]
\begin{scope}
\draw
(0,0)
 +(90:1) node[vwhite,label=right:$v_3$](v3){}
-- +(90+120:1) node[vwhite,label=below:$v_2$](v2){}
-- +(90+240:1) node[vblack,label=below:$v_1$](v1){} -- (v3)
(v3)--   +(90:2) node[vblack,label=right:$u_3$](u3){}
(v2)-- +(90+120:2) node[vblack,label=below:$u_2$](u2){}
(v1)-- +(90+240:2) node[vwhite,label=below:$u_1$](u1){} 
(u3)--++(90:2) node[vwhite,label=above:$x$](x){}
(u2) to[bend left=12] (x) (x)to[bend left=12](u1)
(250:2) node[vwhite](z){}
(u2)--(z) 
(290:2) node[vwhite,label=below:$y$](zz){}
(270:4) node{(a)}
;
\draw[dashed]
(zz)--(u1);
\end{scope}
\begin{scope}[xshift=7cm]
\draw
(0,0)
 +(90:1) node[vblack,label=right:$v_3$](v3){}
-- +(90+120:1) node[vwhite,label=below:$v_2$](v2){}
-- +(90+240:1) node[vwhite,label=below:$v_1$](v1){} -- (v3)
(v3)--   +(90:2) node[vwhite,label=right:$u_3$](u3){}
(v2)-- +(90+120:2) node[vblack,label=below:$u_2$](u2){}
(v1)-- +(90+240:2) node[vblack,label=below:$u_1$](u1){} 
(u3)--++(90:2) node[vwhite,label=above:$x$](x){}
(u2) to[bend left=12] (x) (x)to[bend left=12](u1)
(90:-2) node[vwhite](z){}
(u2)--(z)--(u1)
(u3) ++(220:1)node[vwhite,label=below:$y$](y){}
(270:4) node{(b)}
;
\draw[dashed]
(u3)--(y);
\end{scope}
\begin{scope}[xshift=14cm]
\draw
(0,0)
 +(90:1) node[vwhite,label=right:$v_3$](v3){}
-- +(90+120:1) node[vwhite,label=below:$v_2$](v2){}
-- +(90+240:1) node[vblack,label=below:$v_1$](v1){} -- (v3)
(v3)--   +(90:2) node[vblack,label=right:$u_3$](u3){}
(v2)-- +(90+120:2) node[vblack,label=below:$u_2$](u2){}
(v1)-- +(90+240:2) node[vwhite,label=below:$u_1$](u1){} 
(u3)--++(90:2) node[vwhite,label=above:$x$](x){}
(u2) to[bend left=12] (x) (x)to[bend left=12](u1)
(250:2) node[vwhite,label=below:$y_2$](y2){}
(u2)--(y2) 
(290:2) node[vwhite,label=below:$y_1$](y1){}
(y1)--(u1)
(u3)--++(220:1)node[vwhite,label=below:$y_3$](y3){}
(270:4) node{(c)}
;
\draw[dashed]
(y1)--(y2);
\end{scope}
\begin{scope}[xshift=21cm]
\draw
(0,0)
 +(90:1) node[vwhite,label=right:$v_3$](v3){}
-- +(90+120:1) node[vwhite,label=below:$v_2$](v2){}
-- +(90+240:1) node[vblack,label=below:$v_1$](v1){} -- (v3)
(v3)--   +(90:2) node[vwhite,label=right:$u_3$](u3){}
(v2)-- +(90+120:2) node[vwhite,label=below:$u_2$](u2){}
(v1)-- +(90+240:2) node[vwhite,label=below:$u_1$](u1){} 
(270:2) node[vblack,label=below:$x_{12}$](x12){}
(u1)--(x12)--(u2)
(270+120:2) node[vblack,label=right:$x_{13}$](x13){}
(u1)--(x13)--(u3)
(270+240:2) node[vblack,label=left:$x_{23}$](x23){}
(u3)--(x23)--(u2)
(270:4) node{(d)}
;
\end{scope}
\end{tikzpicture}
\end{center}
\caption{Cases in Lemma~\ref{lemma-free3}.}\label{fig-cases}
\end{figure}

\begin{itemize}
\item The first subcase is that either one of vertices of $U$ has degree two (see Figure~\ref{fig-cases}(a)), or two vertices of $U$ have
a common neighbor distinct from $x$ (see Figure~\ref{fig-cases}(b)), and consequently $|N_G[U]| \leq 9$. 
If $|N_G[U]| \leq 8$, then let $Z=N_G[U]$.
If $|N_G[U]| = 9$, then let $Z= N_G[U] \setminus \{y\}$, where $y$ is a vertex of $N_G[U]\setminus (V(H)\cup U)$
with only one neighbor in $U$.
Note that $|Z|\le 8$ and at most $5$ vertices of $Z$ are $3$-free in $G$, and thus
$m(G-Z) \geq m(G)-5$.  Furthermore, observe that $G[Z]$ contains an independent set $U'$ of size 3 such that $N_G[U']\subseteq Z$.
Hence,
$\alpha(G) \geq \alpha(G') + 3 \geq \frac{n-8}{3}+  \frac{1}{42}(m(G)-5) + 3 >  \frac{n}{3}+  \frac{1}{42}m(G)$.

\item The second subcase is that $|N_G[U]| = 10$, and thus the vertices of $U$ have pairwise distinct neighbors not
contained in $V(H)\cup\{ x\}$.  For $1\le i\le 3$, let $y_i$ denote the neighbor of $u_i$ distinct from $v_i$ and $x$, see Figure~\ref{fig-cases}(c).
Let $G'$ be the graph obtained from $G$ by removing the $8$ vertices of $Y=V(H)\cup U \cup\{x,y_3\}$
and adding the edge $y_1y_2$ if it is not already present (since $G$ is diamond-free, this does not create $K_4$).
All vertices of $G'$ that are 3-free in $G$ are also $3$-free in $G'$, unless they belong to
a triangle, a necklace, or a triangle-dominated $6$-cycle containing the edge $y_1y_2$.  Note that at most one necklace or triangle-dominated $6$-cycle $Q$ of $G'$
contains the edge $y_1y_2$ (only the 6-vertex diamond necklace can intersect another necklace, and this
situation cannot arise since $G$ is diamond-free).  Furthermore, all but at most $9$ vertices of $V(Q)$
are contained in a triangle in $G$, and at most $5$ vertices of $Y$ are $3$-free in $G$.
Consequently $m(G') \geq m(G) - 14$.
Consider any independent set $S$ of $G'$.  Since $y_1y_2\in E(G')$, we can by symmetry
assume that $y_2\not\in S$.  Hence, $S\cup\{v_1,u_2,u_3\}$ is an independent set in $G$.
By the induction hypothesis, this gives
$\alpha(G) \geq \alpha(G') + 3 \geq \frac{n-8}{3}+  \frac{1}{42}m(G') + 3
\geq   \frac{n+1}{3}+  \frac{1}{42}(m(G)-14) = \frac{n}{3}+  \frac{1}{42}m(G)$.
\end{itemize}

Secondly, let us consider the case that there is no common neighbor of all vertices of $U$.
For $1 \leq i < j \leq 3$, let $x_{ij}$ denote a common neighbor of $u_i$ and $u_j$, and observe
that $x_{12}$, $x_{13}$, and $x_{23}$ are three distinct vertices.  Since $G$ does not contain a Havel
necklace as an induced subgraph, the set $\{x_{12},x_{13},x_{23}\}$ is independent in $G$.
Let $G'$ be obtained from $G-N_G[V(H)]$ by identifying $x_{12}$, $x_{13}$, and $x_{23}$ into a new vertex $x$.
Since $G$ does not contain a triangle-dominated $6$-cycle, $G'$ is $K_4$-free.
All vertices of $G'$ that are 3-free in $G$ are also $3$-free in $G'$, unless they belong to
a triangle, a necklace, or a triangle-dominated $6$-cycle containing $x$.  Since $G$ is diamond-free, $x$ is contained in at most one
such subgraph of $G'$, and at most $8$ vertices of this subgraph other than $x$ are $3$-free in $G$.
We have $|V(G)|-|V(G')|=8$, but the vertices of $H$ are not $3$-free in $G$.
Consequently, $m(G') \geq m(G)-14$.
Let $S$ be an independent set in $G'$.
If $x \in S$ then $(S\setminus x) \cup \{x_{12},x_{13},x_{23},v_1\}$ is an independent set of $G$, see Figure~\ref{fig-cases}(d).
If $x \not\in S$ then $S \cup \{u_1,u_2,u_3\}$ is an independent set of $G$.
This gives
$\alpha(G) \geq \alpha(G') + 3 \geq \frac{n-8}{3}+  \frac{1}{42}m(G') + 3
\geq   \frac{n+1}{3}+  \frac{1}{42}(m(G)-14) = \frac{n}{3}+  \frac{1}{42}m(G)$.

\end{proof}

\section{Proofs of the main results}\label{sec-mainproofs}

We need the following corollary of the list-coloring version of Brook's theorem~\cite{borodin1977criterion,erdosrubintaylor1979}.
\begin{lemma}[\cite{borodin1977criterion,erdosrubintaylor1979}]\label{lemma-list}
Let $L$ be a list assignment for a graph $G$ such that $|L(v)|\ge \deg(v)$ for every $v\in V(G)$.
If $G$ is not $L$-colorable, then $G$ contains a clique or an odd cycle $K$ such that all but at most one vertex of $K$
have lists of size $\delta(K)=\Delta(K)$.
\end{lemma}

A set $Z\subseteq V(G)$ is \emph{$\Delta$-free} if either $\Delta>3$ and all vertices of $Z$ are $\Delta$-free, or
$\Delta=3$ and $Z$ can be partitioned so that each part either induces a free diamond or contains only $3$-free vertices.
We say that $G$ is \emph{$K_\Delta$-partitioned} 
if there exists a partition of the vertices of $G$ such that
each part induces a clique of size $\Delta$.  We say that a subset $Z$ of vertices of $G$ is \emph{$\Delta$-profitably nibbled}
if there exists a partition $Z_1,\ldots, Z_r$ of $Z$ such that for $a=1,\ldots,r$, the set $Z_a$ is $\Delta$-profitable in
$G-\bigcup_{i=1}^{a-1} Z_i$ and $|Z_a|\le \max(\Delta+3,10)$.

We are now ready to prove our main decomposition result.

\begin{lemma}\label{lemma-decomp}
There exists an algorithm with time complexity $O(\Delta^2 n)$ that, given as an input 
an integer $\Delta\ge 3$ and an $n$-vertex graph $G$ with $\max(\Delta(G),\omega(G))\le\Delta$, returns a partition of vertices of $G$
to sets $A$, $B$, $C$, and $D$, such that
\begin{itemize}
\item $G[A]$ is $\Delta$-tightly partitioned,
\item $G[B]$ is $K_\Delta$-partitioned and $|B|\le 3\Delta(|C|+|D|)$,
\item $C$ is $\Delta$-profitably nibbled,
\item $D$ is $\Delta$-free in $G-C$, and
\item $\alpha(G)=\alpha(G[B\cup C\cup D])+|A|/\Delta$.
\end{itemize}
\end{lemma}

\begin{figure}
\begin{center}
\tikzset{insep/.style={inner sep=1.7pt, outer sep=0pt, circle, fill},
 noin/.style={inner sep=0pt, outer sep=0pt, circle, fill},}
\begin{tikzpicture}[scale=1.2] 
\def\e{0.3}
\begin{scope}[xshift=0cm,yshift=0cm]
\draw (0,0) node[insep](y){};
\draw (1.5,0) node[insep](b){};
\begin{scope}[xshift=0.75cm,yshift=-0.75cm]
\draw[rotate=0]
(0,\e) node[insep](x1){} -- (0,-\e) node[insep](x2){}
(\e,0) node[insep](z1){} (-\e,0) node[insep](z2){} 
(x1)--(z1)--(x2)--(z2)--(x1);
\end{scope}
\draw(z2)--(y) -- (b) --(z1);
\begin{scope}[xshift=0.75cm,yshift=0.75cm]
\draw[rotate=0]
(0,\e) node[insep](x1){} -- (0,-\e) node[insep](x2){}
(\e,0) node[insep](z1){} (-\e,0) node[insep](z2){} 
(x1)--(z1)--(x2)--(z2)--(x1);
\end{scope}
\draw(z2)--(y) -- (b) --(z1);
\end{scope}
\begin{scope}[xshift=4cm,yshift=0cm]
\draw (0,0) node[insep](y){};
\draw (1.5,0) node[insep](b){};
\begin{scope}[xshift=0.75cm,yshift=-0.75cm]
\draw[rotate=0]
(0,\e) node[insep](x1){} -- (0,-\e) node[insep](x2){}
(\e,0) node[insep](z1){} (-\e,0) node[insep](z2){} 
(x1)--(z1)--(x2)--(z2)--(x1);
\end{scope}
\draw(z2)--(y) -- (b) --(z1);
\draw[rotate=0]
(0.75,0.75) node[insep](x){} (y) -- (x) -- (b);
;
\end{scope}
\end{tikzpicture}
\end{center}
\caption{Set $Y$ for $\Delta=3$ in Lemma~\ref{lemma-decomp}.}\label{fig-profit3}
\end{figure}
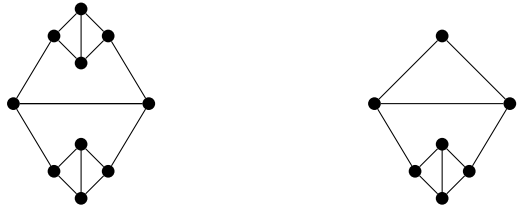

\begin{figure}
\begin{center}
\begin{tikzpicture} 
\draw[rounded corners=20pt] 
(-5,-2) rectangle (5,2)
;
\draw
\foreach \x/\l in {1/A_0,2/A_1,3/A_3,4/B,5/C,6/D}{(-5.8+1.666*\x,0) node(p\x){} ++ (0,-2.5) node{$\l$}};
\foreach \x in {1,2,3,4,5}{
\draw (-5+1.666*\x,-2) -- ++(0,4);
}

\draw[every node/.style={draw,circle,fill,inner sep=0.7pt}]

(p2) 
\foreach \x in {0,1,2,3}{+(45+90*\x:0.25)node(x\x){}}
+(0,0.5) node(e1){} +(0,-0.5) node(e2){} +(0.5,0) node(z1){} +(-0.5,0) node(z2){}
(x0)--(x1)--(x2)--(x3)--(x1)--(z2)--(x2)--(x0)--(x3)--(z1)--(e1)--(z2)--(e2)--(z1)--(x0) (e1)--(e2)

(p1) ++(0,0.8) 
\foreach \x in {0,1,...,7}{+(45*\x:0.5)node(x\x){}} 
(x0)--(x1)--(x2)--(x3)--(x4)--(x5)--(x6)--(x7)--(x0) (x0)--(x2)--(x4)--(x6)--(x0) (x1)--(x3)--(x5)--(x7)--(x1)

(p1) ++(0,-0.8) 
coordinate (XX)
\foreach \x in {0,1,2,3}{{(XX)++(0,0.35)++(45+90*\x:0.25)node(x\x){}}}
(x0)--(x1)--(x2)--(x3)--(x0)--(x2)(x1)--(x3)
(XX) +(0.5,0.35) node(z1){} +(-0.5,0.35) node(z2){}
(x0)--(z1)--(x3)(x2)--(z2)--(x1)
\foreach \x in {0,1,2,3}{{(XX)++(0,-0.35)++(45+90*\x:0.25)node(x\x){}}}
(x0)--(x1)--(x2)--(x3)--(x0)--(x2)(x1)--(x3)
(XX) +(0.5,-0.35) node(y1){} +(-0.5,-0.35) node(y2){}
(x0)--(y1)--(x3)(x2)--(y2)--(x1)
(y1)--(z1) (y2)--(z2)
(y1) to[out=100,in=280] (z2)
(y2) to[out=80,in=260] (z1)

(p3) 
\foreach \x in {0,...,3}{+(45+90*\x:0.4) node(x\x){}--+(45+90*\x:0.7) }  (x0)--(x1)--(x2)--(x3)--(x0)--(x2) (x3)--(x1)

(p4) 
\foreach \x in {0,...,3}{+(90*\x:0.4) node(x\x){}}  (x0)--(x1)--(x2)--(x3)--(x0)--(x2) (x3)--(x1) 
(x0)--+(0.3,0) (x1)--+(0.3,0) (x3)--+(0.3,0) (x2)--+(-0.3,0)

(p5) 
\foreach \x in {0,...,4}{+(90+72*\x:0.5) node(x\x){}}(x0)--(x1)--(x2)--(x4)--(x0)--(x2) (x3)--(x4)--(x1)--(x3)--(x0)
(x2)--++(0,-0.4)node{}  (x3)--++(0,-0.4)node{}

(p6)
node{}
;
\end{tikzpicture}
\end{center}
\caption{Partition of $G$ in Lemma~\ref{lemma-decomp} for $\Delta=4$.}\label{fig-partition}
\end{figure}
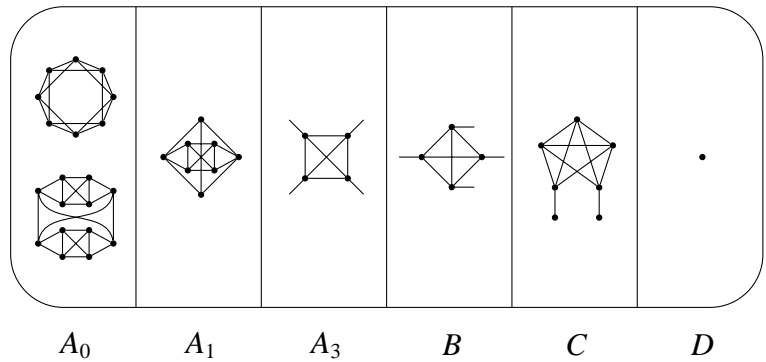

\begin{proof}
See Figures~\ref{fig-partition} and \ref{fig-partition3} for an illustration of the sets that we construct.
Each of the following steps can be easily done in time $O(\Delta^2 n)$:
\begin{itemize}
\item Initialize $C=\emptyset$, and while there exists a set $Y\subseteq V(G)\setminus C$ such that
$G[Y]$ is isomorphic to 
\begin{itemize}
\item if $\Delta \geq 4$, the clique on $\Delta+1$ vertices minus one edge, or
\item if $\Delta = 3$,  two diamond necklaces that share an edge or a diamond necklace sharing
and edge with a triangle (see Figure~\ref{fig-profit3}),
\end{itemize}
add $N_{G-C}[Y]$ to $C$ (note that $|N_{G-C}[Y]|\le \max(\Delta+3,10)$ and $N_{G-C}[Y]$ is $\Delta$-profitable in $C$).
\item If $\Delta\le 5$, find all connected components of $G-C$ of size at most $4\Delta$ that are $\Delta$-tight, and let $A_0$ be the union
of their sets of vertices; otherwise, let $A_0=\emptyset$.
\item If $\Delta=4$, find all subgraphs of $G-(C\cup A_0)$ isomorphic to an extended clique (all these subgraphs are necessarily vertex-disjoint,
and vertex-disjoint from all cliques of size $4$ not fully contained in the subgraph),
and let $A_1$ be the union of their sets of vertices; otherwise, let $A_1=\emptyset$.
\item If $\Delta=3$, find all subgraphs of $G-(C\cup A_0)$ isomorphic to a necklace (all these subgraphs are necessarily vertex-disjoint
and vertex-disjoint from all triangles not fully contained in the subgraph, by the choice of $C$), and let $A_2$ be the union of their sets of vertices; otherwise, let
$A_2=\emptyset$.
\item Let $D$ be the set of vertices of $G-(A_0\cup A_1\cup A_2\cup C)$ that 
\begin{itemize}
\item if $\Delta \ge 4$ are not contained in cliques of size $\Delta$, and
\item if $\Delta = 3$ are not contained in cliques of size $\Delta$ or are contained in diamonds.
\end{itemize}
Let $A'_3=V(G)\setminus (A_0\cup A_1\cup A_2\cup C\cup D)$.
\end{itemize}
Note that since $\Delta(G)\le \Delta$, the choice of $C$ and $D$ ensures that any two cliques of size $\Delta$ in $G[A'_3]$ are vertex-disjoint,
and thus $A'_3$ is $K_\Delta$-partitioned.  Next, perform the following procedure: initialize $B=\emptyset$,
and as long as there exists a clique $K$ of size $\Delta$ in $G[A'_3]$ such that at least $\Delta-1$ vertices of $K$
have neighbors in $B\cup C\cup D$, remove $V(K)$ from $A'_3$ and add it to $B$; let $A_3$ denote the set obtained from
$A'_3$ by performing this procedure.  
Note that each step of the procedure decreases the number of edges between
$A'_3$ and $B\cup C\cup D$ by at least $\Delta-2$, and the number of edges going out of $|C\cup D|$ is at most $\Delta(|C|+|D|)$,
and thus $|B|\le \frac{\Delta^2}{\Delta-2}(|C|+|D|)\le 3\Delta(|C|+|D|)$.

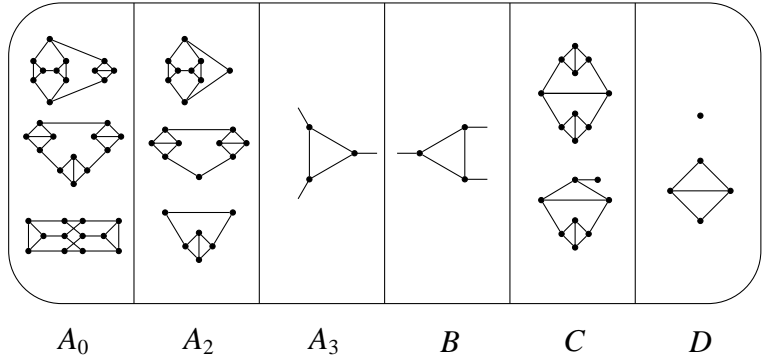
\begin{figure}
\begin{center}
\begin{tikzpicture}
\draw[rounded corners=20pt] 
(-5,-2) rectangle (5,2)
;
\draw
\foreach \x/\l in {1/A_0,2/A_2,3/A_3,4/B,5/C,6/D}{(-5.8+1.666*\x,0) node(p\x){} ++ (0,-2.5) node{$\l$}};
\foreach \x in {1,2,3,4,5}{
\draw (-5+1.666*\x,-2) -- ++(0,4);
}

\draw
(p1)++(0,0) 
node{\begin{tikzpicture}[scale=0.6,every node/.style={draw,circle,fill,inner sep=0.7pt}]
\def\e{0.3}
\draw[rotate=90]
(0,\e) node(x1){} -- (0,-\e) node(x2){}
(\e,0) node(y1){} (-\e,0) node(y2){} 
(x1)--(y1)--(x2)--(y2)--(x1);
\begin{scope}[xshift=1.5cm,yshift=0cm]
\draw[rotate=90]
(0,\e) node(a1){} -- (0,-\e) node(a2){}
(\e,0) node(b1){} (-\e,0) node(b2){} 
(a1)--(b1)--(a2)--(b2)--(a1);
\end{scope}
\begin{scope}[xshift=0.75cm,yshift=-0.75cm]
\draw[rotate=0]
(0,\e) node(x1){} -- (0,-\e) node(x2){}
(\e,0) node(z1){} (-\e,0) node(z2){} 
(x1)--(z1)--(x2)--(z2)--(x1);
\end{scope}
\draw(z2)--(y2)  (y1) -- (b1) (b2)--(z1);
\end{tikzpicture}}
;

\draw (p1)++(0,1.1)
node{\begin{tikzpicture}[scale=0.6,every node/.style={draw,circle,fill,inner sep=0.7pt}]
\draw (0.15,0) node[](a){}
--+(45:0.3) node[](b){}
--+(-45:0.3) node[](c){}--(a)
--++(-0.3,0) node[](x){}
--+(135:0.3) node[](y){}
--+(225:0.3) node[](z){}--(x)
(0,0.7) node[](k){}
(0,-0.7) node[](l){}
(b)--(k)--(y)
(c)--(l)--(z)
(1,0) node[](q1){}
--++(45:0.3) node[](q2){}
--++(-45:0.3) node[](q3){}
--++(45:-0.3) node[](q4){}--(q1)--(q3)
(k)--(q2) (l)--(q4)
;
\end{tikzpicture}
};

\draw (p1) ++(0,-1.1)
node{
\begin{tikzpicture}[scale=0.2,every node/.style={draw,circle,fill,inner sep=0.7pt}]
\draw
(-2,0) node[](z2){}
(-3,-1) node[](z1){}
(-3,1) node[](z3){}
(z1)--(z2)--(z3)--(z1)
(2,0) node[](w2){}
(3,-1) node[](w1){}
(3,1) node[](w3){}
(w1)--(w2)--(w3)--(w1)
\foreach \x in {1,2,3}{
(-0.6,-2+\x) node[](x\x){}--(z\x)
(0.6,-2+\x) node[](y\x){}--(w\x)
}
(x1) -- (y1)--(x2)--(y3)--(x3)--(y2)--(x1)
;
\end{tikzpicture}
};

\draw (p2)++(0,1.1)
node{\begin{tikzpicture}[scale=0.6,every node/.style={draw,circle,fill,inner sep=0.7pt}]
\draw (0.15,0) node(a){}
--+(45:0.3) node[](b){}
--+(-45:0.3) node[](c){}--(a)
--++(-0.3,0) node[](x){}
--+(135:0.3) node[](y){}
--+(225:0.3) node[](z){}--(x)
(0,0.7) node[](k){}
(0,-0.7) node[](l){}
(b)--(k)--(y)
(c)--(l)--(z)
(1,0) node[](q){}
(k)--(q)--(l)
;
\end{tikzpicture}
};

\draw
(p2)++(0,0)
node{\begin{tikzpicture}[scale=0.6,every node/.style={draw,circle,fill,inner sep=0.7pt}]
\def\e{0.3}
\begin{scope}[xshift=3cm,yshift=0cm]
\draw[rotate=90]
(0,\e) node(x1){} -- (0,-\e) node(x2){}
(\e,0) node(y1){} (-\e,0) node(y2){} 
(x1)--(y1)--(x2)--(y2)--(x1);
\begin{scope}[xshift=1.5cm,yshift=0cm]
\draw[rotate=90]
(0,\e) node(a1){} -- (0,-\e) node(a2){}
(\e,0) node(b1){} (-\e,0) node(b2){} 
(a1)--(b1)--(a2)--(b2)--(a1);
\end{scope}
\begin{scope}[xshift=0.75cm,yshift=-0.75cm]
\draw[rotate=0]
(0,0) node(z2){};
\end{scope}
\draw(z2)--(y2)  (y1) -- (b1) (b2)--(z2);
\end{scope}
\end{tikzpicture}}

(p2)++(0,-1.1)
node{\begin{tikzpicture}[scale=0.6,every node/.style={draw,circle,fill,inner sep=0.7pt}]
\def\e{0.3}
\begin{scope}[xshift=6cm,yshift=0cm]
\draw (0,0) node(y){};
\begin{scope}[xshift=1.5cm,yshift=0cm]
\draw (0,0) node(b){};
\end{scope}
\begin{scope}[xshift=0.75cm,yshift=-0.75cm]
\draw[rotate=0]
(0,\e) node(x1){} -- (0,-\e) node(x2){}
(\e,0) node(z1){} (-\e,0) node(z2){} 
(x1)--(z1)--(x2)--(z2)--(x1);
\end{scope}
\draw(z2)--(y) -- (b) --(z1);
\end{scope}
\end{tikzpicture}}
;

\draw[every node/.style={draw,circle,fill,inner sep=0.7pt}]

(p3) 
\foreach \x in {0,...,2}{+(0+120*\x:0.4) node(x\x){}--+(0+120*\x:0.7) }  (x0)--(x1)--(x2)--(x0)

(p4) 
\foreach \x in {0,...,2}{+(180+120*\x:0.4) node(x\x){}}  (x0)--(x1)--(x2)--(x0)
(x0)--++(180:0.3) (x2)--++(0:0.3) (x1)--++(0:0.3)

(p6)++(0,0.5) node{} 

(p6)++(0,-0.5)
\foreach \x in {0,...,3}{+(90*\x:0.4) node(x\x){} }  (x0)--(x1)--(x2)--(x3)--(x0)--(x2)
;

\draw
(p5) ++(0,0.8)
node{\begin{tikzpicture}[scale=0.6,every node/.style={draw,circle,fill,inner sep=0.7pt}]
\def\e{0.3}
\draw (0,0) node(y){};
\draw (1.5,0) node(b){};
\begin{scope}[xshift=0.75cm,yshift=-0.75cm]
\draw[rotate=0]
(0,\e) node(x1){} -- (0,-\e) node(x2){}
(\e,0) node(z1){} (-\e,0) node(z2){} 
(x1)--(z1)--(x2)--(z2)--(x1);
\end{scope}
\draw(z2)--(y) -- (b) --(z1);
\begin{scope}[xshift=0.75cm,yshift=0.75cm]
\draw[rotate=0]
(0,\e) node(x1){} -- (0,-\e) node(x2){}
(\e,0) node(z1){} (-\e,0) node(z2){} 
(x1)--(z1)--(x2)--(z2)--(x1);
\end{scope}
\draw(z2)--(y) -- (b) --(z1);
\end{tikzpicture}}

(p5) ++(0,-0.8)
node{\begin{tikzpicture}[scale=0.6,every node/.style={draw,circle,fill,inner sep=0.7pt}]
\def\e{0.3}
\draw (0,0) node(y){};
\draw (1.5,0) node(b){};
\begin{scope}[xshift=0.75cm,yshift=-0.75cm]
\draw[rotate=0]
(0,\e) node(x1){} -- (0,-\e) node(x2){}
(\e,0) node(z1){} (-\e,0) node(z2){} 
(x1)--(z1)--(x2)--(z2)--(x1);
\end{scope}
\draw(z2)--(y) -- (b) --(z1);
\draw[rotate=0]
(0.75,0.45) node(x){} (y) -- (x) -- (b)
(x)--++(0.5,0)node{}
;
\end{tikzpicture}}
;

\end{tikzpicture}
\end{center}
\caption{Partition of $G$ in Lemma~\ref{lemma-decomp} for $\Delta=3$.}\label{fig-partition3}
\end{figure}

Let $A=A_0\cup A_1\cup A_2\cup A_3$.  Each $\Delta$-tight graph $H$ satisfies $\alpha(H)=|V(H)|/\Delta$, and
since $A$ is $\Delta$-tightly partitioned, clearly $\alpha(G[A])\le |A|/\Delta$.  It follows that
$\alpha(G)\le \alpha(G[B\cup C\cup D])+|A|/\Delta$.

Conversely, consider any independent set $S$ in $G[B\cup C\cup D]$.  Let $L$ be the assignment of lists to vertices of $A$
defined by $L(v)=\{1,\ldots, \Delta\}$ if $v$ has no neighbor in $B\cup C\cup D$ and $L(v)=\{2,\ldots,\Delta\}$ otherwise.
Note that $G[A]$ contains no clique $K$ of size $\Delta$ such that all but at most one vertex of $K$ have lists of size
$\Delta-1$, as otherwise we would have moved $K$ to $B$.  Furthermore, if $\Delta=3$, $G[A]$ contains no odd cycle of length
at least $5$ with all but one vertices having list of size $2$, as otherwise $G[A]$ would not be $3$-tightly partitioned.
Hence, Lemma~\ref{lemma-list} implies that $G[A]$ is $L$-colorable.
Each $\Delta$-tight subgraph $H$ of $G[A]$ satisfies $\alpha(H)\le |V(H)|/\Delta$, and since we are using exactly $\Delta$
colors, $H$ contains $|V(H)|/\Delta$ vertices of each color.  Since $A$ is $\Delta$-tightly partitioned,
the set $S'$ of vertices of color $1$ has size $|A|/\Delta$.  Hence, $S\cup S'$
is an independent set in $G$ of size $|S|+|A|/\Delta$.  Therefore, $\alpha(G)=\alpha(G[B\cup C\cup D])+|A|/\Delta$.
\end{proof}

Next, we show that if $G$ does not contain independent set much larger than the lower bound $n/\Delta$, then
the sets $C$ and $D$ (and thus also $B$) are small.

\begin{lemma}\label{lemma-boundsize}
Let $\Delta\ge 3$ be an integer and $k\ge 0$ a rational number, let $G$ be an $n$-vertex graph with $\max(\Delta(G),\omega(G))\le\Delta$,
and let $A, B, C, D$ be a partition of $V(G)$ as in Lemma~\ref{lemma-decomp}.  If $|C|+|D|\ge 34\Delta^2k$,
then $\alpha(G)\ge n/\Delta+k$ and we can in time $O(\Delta^2n)$ find an induced subgraph $G_0$ of $G$
with $n_0\le 34\Delta^2\lceil k\rceil$ vertices such that $\alpha(G_0)\ge n_0/\Delta+k$.
\end{lemma}
%
%

\begin{proof}
Let $C_1,\ldots,C_r$ be a partition of $C$ showing that $C$ is $\Delta$-profitably nibbled.
For $a=1,\ldots, r+1$ let $G_a=G-\bigcup_{i=1}^{a-1} C_i$.
For $1\le a\le r$, we have $|C_a|\le \Delta+7$ and $\Delta\cdot(\alpha(G_a)-\alpha(G_{a+1}))\ge |C_a|+1$.

Since $\Delta \geq 3$, we have $|C_a|\le \Delta+7 \leq 34 \Delta$.
Hence $|C| \leq r 34 \Delta$.
Moreover, $\Delta\cdot(\alpha(G)-\alpha(G-C))\ge |C|+r$, and thus by Lemmas~\ref{lemma-free} and \ref{lemma-free3},
\begin{align*}
\alpha(G) &\ge \frac{|C|+r}{\Delta} + \alpha(G-C) 
                \ge \frac{|C|+r}{\Delta} + \frac{n-|C|}{\Delta}+\frac{|D|}{34\Delta^2} \\
&= \frac{n}{\Delta}+\frac{r}{\Delta}+\frac{1}{34|D|\Delta^2}
= \frac{n}{\Delta}+\frac{r34\Delta}{34\Delta^2}+\frac{|D|}{34\Delta^2}\\
&\geq  \frac{n}{\Delta}+\frac{|C|+|D|}{34\Delta^2} \geq n/\Delta + k.
\end{align*}

The graph $G_0$ can be defined
as $G_0=G[C_1\cup \ldots\cup C_{\lceil\Delta k\rceil}]$ if $r\ge \Delta k$, and as $G_0=G[C\cup D_0]$ for a set $D_0\subseteq D$
of size $\lceil 34\Delta^2k\rceil-|C|$ otherwise.
This can be computed in $O(\Delta^2n)$ since the partition of $V(G)$ to $A,B,C,D$
can be computed in $O(\Delta^2n)$ by Lemma~\ref{lemma-decomp}.
\end{proof}

Finally, we are ready to prove the results stated in the introduction.

\begin{proof}[Proof of Theorem~\ref{thm-approx}.]
Compute the partition $A, B, C, D$ of $V(G)$ as in Lemma~\ref{lemma-decomp}.
By Lemma~\ref{lemma-boundsize}, $|C|+|D|<34\Delta^2k$,
and since $G[A\cup B]$ is $\Delta$-tightly partitioned, we can set $X=C\cup D$.
\end{proof}

\begin{proof}[Proof of Corollary~\ref{cor-approx}.]
Compute the partition $A, B, C, D$ of $V(G)$ as in Lemma~\ref{lemma-decomp},
let $$k=\frac{|C|+|D|}{34\Delta^2},$$
and return that $\alpha(G)-n/\Delta\ge k$.  That this is a true statement follows from Lemma~\ref{lemma-boundsize}.

On the other hand, since $G[A\cup B]$ is $\Delta$-tightly partitioned, $\alpha(G)\le \alpha(G[C\cup D])+(|A|+|B|)/\Delta\le
|C\cup D|+n/\Delta$, and thus $\alpha(G)-n/\Delta\le |C\cup D|=34\Delta^2k$.
Hence, the algorithm approximates $\alpha(G)-n/\Delta$ up to the factor $34\Delta^2$.
\end{proof}

\begin{proof}[Proof of Corollary~\ref{cor-kernel}.]
Compute the partition $A, B, C, D$ of $V(G)$ as in Lemma~\ref{lemma-decomp}.
If $|C|+|D|\ge 34\Delta^2k$, then return the induced subgraph $G_0$ obtained in Lemma~\ref{lemma-boundsize}.
If $|C|+|D|<34\Delta^2k$, then $|B| \le 3\Delta(|C|+|D|)< 102\Delta^3k$ by the statement of Lemma~\ref{lemma-boundsize}, and $|B\cup C\cup D|<114\Delta^3 k$.
By Lemma~\ref{lemma-decomp}, the graph $G_0=G[B\cup C\cup D]$ satisfies $\alpha(G_0)-n_0/\Delta=\alpha(G)-n/\Delta$.
\end{proof}

\section*{Acknowledgements}

We would like to thank Andrea Munaro, who found an error in a previous version of this paper.

\bibliographystyle{plainurl}
\bibliography{deltaker}

\end{document}